\documentclass[review]{elsarticle}

\usepackage{amsmath,amssymb,amsfonts, amsthm}
\usepackage{dsfont}
\usepackage{algorithmic}
\usepackage{graphicx}
\usepackage{textcomp}
\usepackage{xcolor}
\usepackage{mathtools}
\usepackage{url}
\usepackage{xspace}
\usepackage{bbm}
\usepackage{accents}
\usepackage{graphicx}
\usepackage{tikz}
\usepackage{pgfplots}
\usepackage[textwidth=1.5in]{todonotes}
\usepackage{caption}
\usepackage{subcaption}
\usepackage{pgfplots}
\usepgfplotslibrary{groupplots}
\pgfplotsset{compat=1.12}

\newtheorem{assump}{Assumption}

\newtheorem{remark}{Remark}
\newtheorem{example}{Example}

\newtheorem{lemma}{Lemma}
\newtheorem{theorem}{Theorem}

\newtheorem{proposition}{Proposition}

\definecolor{mycolor4}{RGB}{230,97,1}
\definecolor{mycolor3}{RGB}{178,171,210}
\definecolor{mycolor2}{RGB}{255,153,7}
\definecolor{mycolor1}{RGB}{94,60,153}

\DeclareMathOperator*{\argmin}{arg\,min}
\DeclareMathOperator*{\E}{\mathbb{E}}

\renewcommand{\P}{\mathbb{P}}

\newcommand{\vdiff}[3]{\Delta^{#1}_{#2}{#3}}

\DeclarePairedDelimiter\abs{\lvert}{\rvert}%

\newcommand{\ubar}[1]{\underaccent{\bar}{#1}}

\begin{document}
\begin{frontmatter}
\title{Sensitivity to User Mischaracterizations in Electric Vehicle Charging}
\author[1]{Cesar Santoyo}
\ead{csantoyo@gatech.edu}
\author[2]{Gustav Nilsson}
\ead{gustav.nilsson@epfl.ch}
\author[1,3]{Samuel Coogan}
\ead{sam.coogan@gatech.edu}

\address[1]{School of Electrical and Computer Engineering, Georgia Institute of Technology, Atlanta, 30332, USA.}
\address[2]{School of Architecture, Civil and Environmental Engineering, École Polytechnique Fédérale de Lausanne (EPFL), 1015 Lausanne, Switzerland.}
\address[3]{School of Civil and Environmental Engineering, Georgia Institute of Technology, Atlanta, 30332, USA.}
\fntext[thanks]{C. Santoyo was supported by the NSF Graduate Research Fellowship Program under Grant No. DGE-1650044. This work was supported in part by the NSF under NSF grant 1931980. A preliminary form of this work appears in \citep{santoyo2021sensitivity}. }

\begin{abstract}
In this paper, we consider electric vehicle charging facilities that offer various levels of service, i.e., charging rates, for varying prices such that rational users choose a level of service  that minimizes the total cost to themselves including an opportunity cost that incorporates users' value of time. In this setting, we study the sensitivity of the expected  occupancy at the facility to mischaracterizations of user profiles, e.g., user's value of time, and uncharacterized heterogeneity, e.g., user charging level possibilities, or the likelihood of early departure.  For user profile mischaracterizations, we first provide a fundamental upper bound for the difference between the expected occupancy under any two different distributions on a user's impatience (i.e., value of time) that only depends on the minimum and maximum charging rate  offered by the charging facility. Next, we consider the case when a user's impatience is a discrete random variable and study the sensitivity of the expected occupancy to the probability masses and attained values of the random variable. We show that the expected occupancy varies linearly with respect to the probability masses and is piecewise constant with respect to the attained values. 
Furthermore, we study the effects on the expected occupancy from the occurrence of heterogeneous user populations. In particular, we quantify the effect on the expected occupancy from the existence of sub-populations that may only select a subset of the offered service levels. Lastly, we quantify the variability of early departures on the expected occupancy. These results  demonstrate how the facility operator might design prices such that the expected occupancy does not vary much under small changes in the distribution of a user's impatience,  variable and limited user service needs, or uncharacterized early departure, quantities which are generally difficult to characterize accurately from data. We further demonstrate our results via examples.
\end{abstract}
\end{frontmatter}

\section{Introduction}
Government incentives coupled with market-driven cost reductions \citep{sierzchula2014influence}
of electric vehicles (EVs) have catalyzed EV adoption such that by 2040 it is projected that 58\% of global new vehicle sales will be EVs \citep{bloombergnef}. This increase in EV adoption 
calls for significant investments in commercial charging infrastructure \citep{hoover2021charging, engel2018charging}. As commercial charging facilities become more abundant, there is an increased need to investigate EV charging facility management schemes since it is in the interest of charging facility operators to understand the effects of their operational models on users.

There have been efforts to study the problem from various perspectives such as scheduling and wholesale perspective \citep{sassi2017electric, alizadeh2018retail, le2016optimal} and including a pricing model approach \citep{santoyo2020multi, santoyo2020pricing, zhao2013pricing}. Additional past works have focused on studying EV charging within a utility provider framework where the utility provider and charging facility are separate entities whose pricing engenders specific system-wide behavior \citep{alizadeh2016optimal}. Furthermore, the paper \citep{liu2016distribution} presents a location-based pricing scheme and analyzes its effects on system-wide congestion. The paper~\citep{bae2011spatial} considers a spatiotemporal, queuing theoretic model for rapid charging facilities to predict charging demand when the user arrival rate is not known a priori. The paper~\citep{liu2021electric} models EV charging within a queuing framework to formulate an equilibrium assignment model.  The paper~\citep{flath2012revenue} studies the EV charging revenue problem by adapting capacity control mechanisms from asset revenue management to allocate charging capacity. Lastly, the paper \citep{zhao2013pricing} studies the revenue maximization problem at plug-in hybrid charging stations by studying the equilibria of customer subscription dynamics. Each of these previous works relies on specific modeling assumptions; however, none of the aforementioned papers perform a sensitivity analysis of their respective models to understand how their models perform when the required information is not exactly known. The need for understanding EV charging model sensitivity is the main motivator for the present work.

In the present paper, we focus on the charging facility service level model that is originally presented in the paper \citep{santoyo2020multi}.  Here, users arrive at random times with a collection of random parameters whose respective distributions are assumed to be known. Specifically, the user parameters are the user's energy demand and their value of time. Furthermore, in this model, users are presented with a discrete and finite collection of charging rates and energy prices. Upon arrival, users choose the charging rate and price that minimizes the total cost to themselves which includes their opportunity cost. Each charging rate is associated with a specific service level pricing function that defines the total cost to the user of choosing a particular charging rate as a function of the user's demand parameters.

Since users present a collection of random demands and arrivals to the charging facility, we study the expected (mean) occupancy. For the pricing schemes derived in the papers \citep{santoyo2020multi, santoyo2020pricing}, explicit formulas for the expected value are derived under the assumption that the distributions of user arrivals and their parameters are known. Specifically, these formulas require knowing the distribution of the user's value of time, assuming all users have the same usage profile, and that users will remain at the charging facility until charge completion. 

In practice, given the wide availability of data on EVs and consumer habits, it is reasonable to assume that a charging facility can obtain good estimates of quantities that can be explicitly measured such as user arrival times and user's energy demands. While recent studies attempt to quantify behavioral factors \citep{asensio2021field}, in practice, obtaining an accurate characterization of a user's impatience remains challenging. Furthermore, the existence of heterogeneous charging populations, e.g., a subset of users who can only use one of the offered charging levels, and populations who depart early affects the estimates of expected occupancy.

We are motivated to study the sensitivity of the expected occupancy to such mischaracterizations in the user populations. First, we derive a worst-case error bound for the expected occupancy. Second, we derive the gradient of the expected occupancy that describes how the expected occupancy changes with a given characterization of the user's value of time.  Here, the user's impatience is considered to be a discrete random variable.

The present work extends the paper \citep{santoyo2021sensitivity} by considering the sensitivity of the expected occupancy to heterogeneous user populations. In particular, the present work derives a gradient expression for when user populations do not have the same charging capabilities, e.g., some user's vehicles may only use a specific service level, or some older EVs may not have the technology required for high-speed charging. This extension addresses the practical cases where certain vehicles are restricted to specific charger types \citep{afdc_energygov}, hence affecting the expected occupancy at a charging facility. Lastly, we consider the occurrence of early departures at a charging facility. The service level model in the paper \citep{santoyo2020pricing} works under the assumptions that users remain at the charging facility until they receive a full charge.

This paper is organized as follows: Section \ref{sect: problem_formulation} presents the model formulation for the pricing model. Section \ref{sect: main_results} presents the sensitivity analysis results for mischaracterized user profiles and uncharacterized heterogeneity in the user population. Section \ref{sect: case_study} details a numerical study and Section \ref{sect: conclusion} presents the conclusions.

\subsection{Notation}
For an indexed set of variables $\{x^k\}$, we let $\Delta^i_j x$ denote the difference between the variable with index $i$ and $j$, i.e., $\vdiff{i}{j}{x} = x^i - x^j$.
When considering a collection of independent and identically distributed (i.i.d) random variables indexed by subscripts, we use non-subscript variables when referring to properties that hold for any of the i.i.d random variables. For example, $\mathbb{E}[x]$ is the expectation of each i.i.d random variable $x_j$, $j$ belonging to some index set. For some set $\mathcal{A}$, define $\mathcal{P}(\mathcal{A})$ to be the power set of $\mathcal{A}$ minus the empty set $\varnothing$, i.e, the set of all subsets of $\mathcal{A}$ excluding the empty set.

\section{Problem Formulation}
\label{sect: problem_formulation}
In this section, we present a  pricing model for EV charging that was initially introduced in the papers \citep{santoyo2020multi, santoyo2020pricing} and captures practical uses of charging level equipment \citep{afdc_energygov}. We consider a \textit{defined service level model} where users directly choose from a discrete set of charging rates and prices upon arrival at the charging facility; a user pays a higher price for a faster charge rate. A rational user chooses a cost-minimizing charge rate depending on the amount of charge required for their EV, the prices and rates set by the charging facility, and their impatience factor, i.e., their value of time.

At this facility, a user $j$ arrives at some time $\tau_j$ (in hr.) with charging demand $x_j$ (in kWh), and an impatience factor~$\alpha_j$ (in \$/hr.). Throughout the paper, we make the following assumption about the aforementioned variables.

\begin{assump}[Users]
\label{assump: rv_arrival}
User arrivals at the charging facility is a Poisson process with parameter $\lambda$ (in EVs/hr.). Individual charging demand $x_j$ and the impatience factor $\alpha_j$ for each user $j$ are random variables which are independent and identically distributed (i.i.d). In particular, $x_j$ is a continuous random variable with support $[x_\text{min}, x_\text{max}]$ for some $0 < x_\text{min}< x_\text{max}$. Furthermore, $\alpha_j$ is a discrete random variable with $M$ possible values whose probability mass function $p_A(\alpha; p, a)$ has a probability mass vector $p = [p_1, \dots, p_M]^\top$ corresponding to 
the impatience value vector $a = [a_1, \dots, a_M]^\top$ such that $\mathbb{P}(\alpha_j=a_i)=p_A(a_i; p, a) = p_i$ for each $i$.  
\end{assump}

When using the probability operator $\P(\,\cdot\,)$ for $\alpha_j$ it is understood that this probability is computed with some probability mass vector $p$ and impatience category vector $a$. By assuming $\alpha_j$ are i.i.d discrete random variables, we assume the population of users is divided into a finite number of impatience categories; for example, each user may be patient with a low value of $\alpha$ or impatient with a high value of $\alpha$. 

\begin{table}
        \centering
        \caption{User Parameter Definitions}
        \label{table:userparam}
        \begin{tabular}{cllc} 
        \textbf{Var.} & \textbf{Parameter} & \textbf{Unit} & \multicolumn{1}{l}{\textbf{Range}} \\ \hline 
        $j$ & user index & - & - \\
        $\tau_j$ & arrival time & hr. & - \\
        $x_j$ & user demand & kWh & $[x_\text{min}, x_\text{max}]$\\ 
        $\alpha_j$ & impatience factor & \$/hr. & $\{a_1, \dots, a_M\}$\\
        $r_j$ & charging rate & kW & $(0, R^\text{max}]$\\
        [0.1em] \hline
        \end{tabular}
        \label{table: facilityparam}
\end{table}

The user parameters, their respective units, and upper and lower bounds are summarized in Table~\ref{table:userparam}.
The charging facility offers $L$ service levels such that $\mathcal{L} = \{1, \dots, L\}$ is the set of offered service levels. Each service level $\ell \in \mathcal{L}$ corresponds to a distinct charging rate $R^\ell>0$ (in kW) and price $V^\ell > 0$ (in \$/kWh) that is the cost per unit energy for the service level. Thus, user $j$ with energy demand $x_j$ pays $x_jV^\ell $ (in \$) to receive a full charge over the time horizon $x_j/R^\ell$ (in hr.) when choosing service level $\ell$.

The parameters related to the charging facility under a discrete pricing model are listed in Table~\ref{table:facilityparam}. To distinguish the parameters related to the charging facility from those related to the users, the charging facility parameters are upper case and indexed by a superscript, while the parameters for the users are lower case and indexed by a subscript $j$.

\begin{table}
    \centering
    \caption{Parameter Definitions for the Charging Facility}
     \label{table:facilityparam}
    \begin{tabular}{cllc}
    \textbf{Var.} & \textbf{Parameter} & \textbf{Unit} & \multicolumn{1}{l}{\textbf{Range}} \\ \hline 
    $\ell$ & service level & - & $\{1,\ldots, L\}$ \\
    $V^\ell$ & price per unit of energy & \$/kWh & -\\
    $R^\ell$ & charging rate & kW & $(0, R^\text{max}]$\\
    [0.1em]  \hline
    \end{tabular}
\end{table}

\begin{assump}[Model Charging Rates]
\label{assump: ordering_of_func}
Among $L$ service levels offered by the charging facility, a higher charging rate is more costly, i.e., if $R^i > R^k$ then $V^i > V^k$. Moreover, charging rates and prices are distinct so that $R^i \neq R^k$ for all $i \neq k$. Lastly, and without loss of generality, the charging facility's pricing functions are enumerated such that $V^1<V^2<\ldots <V^L$ and therefore $R^1<R^2<\ldots<R^L$. 
\end{assump}

A user can therefore pay less by choosing a slower charge rate but must balance this with their impatience. In particular, the  total cost faced by a user arriving at the charging facility with impatience factor $\alpha_j$, charging demand $x_j$, and who chooses service level $\ell$, is 
\begin{align}
    \label{eq: cost_func_special_case}
    g_\ell(x_j, \alpha_j) =  x_jV^\ell  + \alpha_j\frac{x_j}{R^\ell} \, .
\end{align}

Individual users choose a service level at a charging facility which minimizes their total cost of charging factoring in their impatience. To that end, let $S(x_j, \alpha_j): [x_\text{min}, x_\text{max}] \times \{a_1, \dots, a_M \} \to \{1, \ldots, L\}$ be defined by 
\begin{align}
    \label{eq: level_func}
    S(x_j, \alpha_j) =  \argmin\limits_{\ell \in \{1, \dots, L \}} g_\ell(x_j, \alpha_j)\,.
\end{align} 
Then, a rational user $j$ chooses service level $S(x_j, \alpha_j)$ in order to minimize their total cost as formalized in the later stated assumption. 

\begin{assump}[Impatience Ambiguity]
\label{assump: ratio_cannot_equal_a_i}
A charging facility offers $L$ service levels with price per unit energy $V^\ell$ and charging rate $R^\ell$ according to Assumption \ref{assump: ordering_of_func} such that $\Delta^i_k V/\Delta^k_i \bar{R} \neq a_m$ for all $k,i$ for any $m$. 
\end{assump}

Assumption \ref{assump: ratio_cannot_equal_a_i} holds generically and avoids the scenario where two charging levels are equally attractive to a user from the perspective of \eqref{eq: cost_func_special_case}, i.e., the minimization in \eqref{eq: level_func} does not have a unique minimizer.

For notational convenience, we also define the values $r_j$ to be the charging rate and cost per unit of energy chosen by user $j$ after solving~\eqref{eq: level_func}, i.e., $r_j = R^{S(x_j, \alpha_j)}$, as indicated in Table~\ref{table:userparam}. Observe that the user charging times $x_j/r_j$, being uniquely determined by $x_j$ and $\alpha_j$,
constitute a collection of independent and identically distributed random variables. Furthermore, this means the time a user spends at the charging location is $x_j/r_j$ where this is the time for a user to receive a full charge based on their chosen service level.

\begin{assump}[Users are Rational]
\label{assump: timeatcharger}
Each user chooses a charging rate according to~\eqref{eq: level_func} and leaves the charging facility once they have satisfied their charging demand.
Thus, user $j$ occupies a charger at the facility during the time interval $[\tau_j, \tau_j+ x_j/r_j]$ to receive full charge $x_j$.  
\end{assump}  

\begin{remark}
Note that the usage of Assumption \ref{assump: timeatcharger} is explicit particularly in the case where it is assumed that users remain at the charging facility to receive the requested charge $x_j$. In the present work, we do consider the practical scenario of early departures which is detailed at the end of Section \ref{sect: main_results}. We introduce the possibility of early departures distinctly from the problem formulation to distinguish the extension from the previously developed service level model. 
\end{remark}
Let the occupancy at the charging facility be defined as $\eta$. Note that $\eta$ is also the number of actively charging users.

\section{Main Results}
\label{sect: main_results}
At a given charging facility with pricing functions of the form of \eqref{eq: cost_func_special_case}, users arrive at random times with random parameters to ultimately make a service level choice that minimizes the cost to themselves by solving \eqref{eq: level_func}. Knowing the probability distribution of user arrivals and their respective parameters enables EV charging facility operators to analyze the system-wide behavior at their facility via the expected occupancy.

To obtain an expression for the occupancy, recall Lemma \ref{lemma: probability_choose_level_special_case} below from the paper \citep{santoyo2020pricing}, which details how a user $j$ arriving with charging demand $x_j$ and impatience factor $\alpha_j$ chooses a specified service level while facing pricing functions of the form of \eqref{eq: cost_func_special_case}. We consider the case where  consequently, Lemma \ref{lemma: probability_choose_level_special_case} also provides an analytical expression for the probability mass function (PMF) for user's choice of charging rate.  

\begin{lemma}[Corollary 1 of \citep{santoyo2020pricing}]
    \label{lemma: probability_choose_level_special_case}
     Under assumptions \ref{assump: rv_arrival}, \ref{assump: ordering_of_func}, \ref{assump: ratio_cannot_equal_a_i}, and  \ref{assump: timeatcharger}, consider the set of $L$ functions of two independent RVs $\big\{g_\ell(x_j,\alpha_j) \big\}^L_{\ell = 1}$ where each $g_\ell$ is as defined in \eqref{eq: cost_func_special_case}. Then,  for $k \in \{1, \ldots, L \}$,
    \begin{equation*}
         \P\bigl(S(x_j, \alpha_j) = k\bigr) = \P \left(\ubar{\alpha}^k < \alpha_j< \Bar{\alpha}^k \right)
    \end{equation*}
    where $\ubar{\alpha}^1 = -\infty$ and $\bar{\alpha}^L = +\infty$ otherwise
    \begin{align}
        \label{eq:alpha_1}
        \Bar{\alpha}^k &=  \ \min_{k < i}\frac{\vdiff{i}{k}{V}}{\vdiff{k}{i}{\bar R}} \,,\\
        \label{eq:alpha_2}\ubar{\alpha}^k &=  \max_{i<k} \frac{\vdiff{i}{k}{V}}{\vdiff{k}{i}{\bar R}}\,.
    \end{align}
    Furthermore, the charging rate $r_j$ chosen by each user $j$ is a 
    discrete random variable each with PMF 
    \begin{equation}
        \label{eq: pmf_special_case}
        p_{r}(r; p, a) = 
        \begin{cases}       
            \P \left(\ubar{\alpha}^1 < \alpha_j < \Bar{\alpha}^1 \right) & \text{if } r = R^1 \,, \\ 
            \qquad \quad  \vdots &  \\
            \P \left(\ubar{\alpha}^L < \alpha_j < \Bar{\alpha}^L \right) & \text{if } r = R^L \,.
        \end{cases}
    \end{equation}
\end{lemma}

Lemma \ref{lemma: probability_choose_level_special_case} demonstrates that the probability of choosing a particular price per unit of energy $V^k$ and charging rate $R^k$ 
solely depends on the likelihood that a user~$j$'s impatience factor $\alpha_j \sim p_A(\alpha;p,a)$ falls within the interval $(\ubar{\alpha}^k , \Bar{\alpha}^k)$. 

In the remainder of this section, we derive a worst-case error bound on the expected occupancy at charging facility. Furthermore, we focus on providing insight on the deviation presented in Proposition \ref{proposition:maxdiff_expectation}, given some knowledge on how the distribution of impatience factors deviates. First, we focus on deviations in $p$, the vector of probabilities that a user will have a particular impatience factor. Then, we characterize deviations in $a$, the vector of possible impatience factors for the population of users. In addition to mischaracterized user profiles, we consider the possibility of uncharacterized heterogeneity in the arriving users such that there are restrictions on which of the offered service levels can be chosen. Lastly, we discuss the effects of early departures on the expected occupancy at the charging facility. 

\subsection{Worst-Case Occupancy Error Bound}
\label{subsect: worst_case_err_bound}
A charging facility with pricing functions of the form~\eqref{eq: cost_func_special_case} experiences user arrivals with i.i.d impatience factors that are distributed with $p_A(\alpha; p,a )$. Any expected value that is dependent on $p_A(\alpha; p,a )$ is written as $\E[\,\cdot\, ;p,a]$.
Specifically, we write the expected occupancy at a charging facility as 
$\E[\eta; p, a] = \lambda \E[x/r;p, a] = \lambda \E[x] \E[\frac{1}{r}; p, a]$. We break up $\E[\eta; p, a]$ like this because, as is seen in Lemma \ref{lemma: probability_choose_level_special_case}, the probability of choosing a particular charge rate is independent of the user demand $x_j$. 

Since, in this paper, we are interested in studying how the expected occupancy changes when users' impatience changes, we start by stating a general upper bound on the deviation of the expected occupancy.

\begin{proposition}
\label{proposition:maxdiff_expectation}
Consider a charging facility operating under Assumptions~\ref{assump: rv_arrival}, \ref{assump: ordering_of_func}, \ref{assump: ratio_cannot_equal_a_i}, and~\ref{assump: timeatcharger}. 
Define $\eta$ to be the occupancy possibly with two different probability mass functions for the impatience factor, $p_A(\alpha; p, a)$ and $p_{A}(\alpha, \tilde{p}, \tilde{a})$, then  
\begin{align*}
    \abs*{\E[\eta;p,a] - \E[\eta;\tilde{p},\tilde{a}]} \leq \lambda \E[x]\left(\frac{1}{R^1} - \frac{1}{R^L}\right).
\end{align*}
Moreover, there exists probability mass functions such that the bound is tight.
\end{proposition}
\begin{proof}
As observed earlier, the choice of charging rate is independent of the user's demand, hence
\begin{align*}
    \E[\eta;p,a] = \lambda \E[x] \E[\frac{1}{r}; p,a] \, .
\end{align*}
Let $p_r(r;p, a)$ and $p_{r}(r; \tilde{p}, \tilde{a})$ denote the corresponding distributions for the charging rates, according to Lemma~\ref{lemma: probability_choose_level_special_case}. We then obtain
\begin{equation}
     \abs*{\E[\eta;p,a] - \E[\eta; \tilde{p}, \tilde a]} =  
     \lambda \E[x] \abs*{\sum^L_{\ell = 1} \frac{p_r(R^\ell; p,a )}{R^\ell}
    - \sum^L_{\ell = 1} \frac{p_r(R^\ell;\tilde{p}, \tilde{a})}{R^\ell}}\,. \label{eq:prop1proof1}
\end{equation}
It can be seen that it is possible to choose distributions of $\alpha$ such that $p_r(R^1; p,a) = 1$ or $p_r(R^L; p,a) = 1$, and  $p_r(R^1; \tilde p,\tilde a) = 1 - p_r(R^1; p,a)$, $p_r(R^L; \tilde p, \tilde a) = 1 - p_r(R^L; p,a)$, which yields the maximum difference in~\eqref{eq:prop1proof1}. 
\end{proof}

\medskip 

Proposition~\ref{proposition:maxdiff_expectation} demonstrates that the difference between the expected occupancy computed with the true and the mischaracterized PMF of $\alpha_j$ is upper bounded. More specifically, assuming a correct characterization of $\lambda$ and $\E[x]$, this worst-case upper bound is driven by the difference of the inverse of the fastest and slowest charging rate.

\subsection{Mischaracterized User Profiles}
A charging facility knows that the expected occupancy varies with $p$ and $a$ of the user impatience. Hence, in this section, we derive results that enable a charging facility to quantify the variability of the expected occupancy with respect to $p$ or $a$.

\label{subsect: mischaracterized_user_profille}
\begin{theorem}
\label{theorem: E_derivative_P}
Consider a charging facility operating 
under Assumptions \ref{assump: rv_arrival}, \ref{assump: ordering_of_func}, \ref{assump: ratio_cannot_equal_a_i}, and \ref{assump: timeatcharger} with $L$ pricing functions of the form of \eqref{eq: cost_func_special_case}.  Then we have that, 
\begin{align}
    \label{eq: E_derivative_P}
    \nabla_p \E[\eta; p,a] = \lambda \E[x] \begin{bmatrix}
    \sum_{\ell = 1}^{L} \frac{\mathds{1}_{\ell}(a_1)}{R^\ell} \\ 
    \vdots \\ 
    \sum_{\ell = 1}^{L} \frac{\mathds{1}_{\ell}(a_M)}{R^\ell}
    \end{bmatrix},
\end{align}
where $\mathds{1}_{\ell}(a_i) = 1$ if $\ubar{\alpha}^\ell < a_i < \bar{\alpha}^\ell$ and $\mathds{1}_{\ell}(a_i) = 0$ otherwise, and we recall that $\E[\eta; p,a]$ is the expected occupancy where user's impatience factors are distributed with $p_A(\alpha; p, a)$.
\end{theorem}

\begin{proof}
Given a charging facility operating under Assumptions \ref{assump: rv_arrival}, \ref{assump: ordering_of_func}, \ref{assump: ratio_cannot_equal_a_i}, and \ref{assump: timeatcharger} with $L$ pricing functions of the form of \eqref{eq: cost_func_special_case}. Recalling the PMF of $\alpha_j$, we note that the probability of choosing a specified charge rate $\ell$ has an equivalence where
\begin{align*}
    \P \left(\ubar{\alpha}^\ell < \alpha_j < \Bar{\alpha}^\ell \right) = \sum^{M}_{m = 1} p_m \mathds{1}_{\ell}(a_m) \,.
\end{align*}
 Recall that $\E[\eta; p,a] = \lambda \E[x] \E[\frac{1}{r}; p,a]$. Furthermore, we expand on $\E[\frac{1}{r}; p, a]$ such that
\begin{align*}
    \E\left[\frac{1}{r}; p,a\right] &= \sum_{\ell = 1}^{L} p_r(R^\ell; p, a)\frac{1}{R^\ell} \\
     &= \sum_{\ell = 1}^{L} \P \left(\ubar{\alpha}^\ell < \alpha_j < \Bar{\alpha}^\ell \right) \frac{1}{R^\ell} \\
     &= \sum_{\ell = 1}^{L} \left( \sum^{
    M}_{m = 1} p_m \mathds{1}_{\ell}(a_m)\right)\frac{1}{R^\ell}\,.
\end{align*}
Given the prior substitutions, we compute the gradient of $\E[\eta; p, a]$ with respect to $p$ leading to \eqref{eq: E_derivative_P}. 
\end{proof}

\medskip

From \eqref{eq: E_derivative_P} in Theorem \ref{theorem: E_derivative_P} we see that the gradient with respect to $p$ of $\E[\eta;p,a]$ is constant; hence, $\E[\eta;p,a]$ varies linearly with $p$. A direct corollary of Theorem \ref{theorem: E_derivative_P} is that of the gradient of the expected occupancy with respect to the difference between a true and mischaracterized probability mass where the expected occupancy also varies linearly. 
In addition to Theorem \ref{theorem: E_derivative_P}, a charging facility is interested in how the expected occupancy varies with specific impatience values. This is formalized in the following theorem.

\begin{theorem}
\label{theorem: E_derivative_a}
Consider a charging facility operating 
under Assumptions \ref{assump: rv_arrival}, \ref{assump: ordering_of_func}, \ref{assump: ratio_cannot_equal_a_i} ,and \ref{assump: timeatcharger} with $L$ pricing functions of the form of \eqref{eq: cost_func_special_case}. Recall that $\E[\eta; p,a]$ is the expected occupancy where user's impatience factors are distributed with $p_A(\alpha; p, a)$. Then, 
\begin{itemize}
    \item[1)] 
    For all $p, a$ such that for every $a_i \in (\ubar{\alpha}^k, \bar{\alpha}^k)$ for some $k>0$, it holds that
\begin{align*}
    \nabla_a \E[\eta; p,a] = 0 \,.
\end{align*}
    \item[2)]  For all $a$, for all $i$, and for all $k<L$,
\begin{multline}
    \label{eq: E_derivative_a}
    \lim_{\tilde{a}_i \uparrow \bar{\alpha}^k} \E[\eta;p,\tilde{a}] - \lim_{\tilde{a}_i \downarrow \bar{\alpha}^k} \E[\eta;p,\tilde{a}] = \\ \lim_{\epsilon \rightarrow 0^+} \lambda \E[x] \left(\sum_{\ell = 1}^{L} \frac{p_i}{R^\ell} \left(   \mathds{1}_{\ell}(\bar{\alpha}^k -\epsilon) - \mathds{1}_{\ell}(\bar{\alpha}^k + \epsilon) \right)  \right) \,,
\end{multline}
where 
$\tilde{a}_j = a_j$ for all $j \neq i$, and $\mathds{1}_{\ell}(a_m ) = 1$ if $\ubar{\alpha}^\ell < a_m < \bar{\alpha}^\ell$ and $\mathds{1}_{\ell}(a_m) = 0$ otherwise.
\end{itemize}
\end{theorem}
\begin{proof}
To prove the first part of Theorem \ref{theorem: E_derivative_a}, recall $\E[\eta; p,a] = \lambda \E[x] \E[\frac{1}{r}; p, a]$. Furthermore, 
\begin{align*}
    \E\left[\frac{1}{r}; p,a\right]
     &= \sum_{\ell = 1}^{L} \left( \sum^{
    M}_{m = 1} p_m \mathds{1}_{\ell}(a_m)\right)\frac{1}{R^\ell}
\end{align*}
where $\mathds{1}_{\ell}(a_m ) = 1$ if $\ubar{\alpha}^\ell < a_m < \bar{\alpha}^\ell$ and $\mathds{1}_{\ell}(a_m) = 0$ otherwise for any $m \in \{1, \dots, M\}$.
Since we have $p, a$ such that $a_i \in \left(\ubar{\alpha}^k, \bar{\alpha}^k\right)$ for some $k$ for all $i$ then the quantity $\E\left[\frac{1}{r}; p,a\right]$ is constant with respect to $a$. This is because $a_i \in (\ubar{\alpha}^k, \bar{\alpha}^k)$ and since the probability mass remains in the interval there is no change to $\E[\eta; p,a]$. Hence, $\nabla_a \E[\eta; p,a] = 0$.

To prove the second part of Theorem \ref{theorem: E_derivative_a}, we analyze the difference $ \E[\eta; p,a^-] - \E[\eta; p, a^+]$ where $\E[\eta; p,a^-]=\lim_{\tilde{a}_i \uparrow \bar{\alpha}^k} \E[\eta;p,\tilde{a}]$ and  $\E[\eta; p,a^+]=\lim_{\tilde{a}_i \downarrow \bar{\alpha}^k} \E[\eta;p,\tilde{a}]$ and  $\tilde{a}_j  = a_j$ for some $j$. From before, realize that $\E[\eta; p,a^-] = \lambda \E[x] \E[\frac{1}{r};p, a^-]$ and $\E[\eta; p,a^+] = \lambda \E[x] \E[\frac{1}{r};p, a^+]$. Substituting the summation form of $\E[\frac{1}{r};p, a^-]$ and $\E[\frac{1}{r};p, a^+]$ leads to \eqref{eq: E_derivative_a}.
\end{proof}


\medskip

Theorem~\ref{theorem: E_derivative_a} first states that $\E[\eta; p,a] = E[\eta;p,\tilde{a}]$ for all $\tilde{a}$ such that $a_i \neq \tilde{a}_i$ and $a_i, \tilde{a}_i \in (\bar{\alpha}^k, \ubar{\alpha}^k)$ for some $k$, i.e., the expected occupancy does not vary when an impatience value $a_i$ is changed  within a given $(\bar{\alpha}^k, \ubar{\alpha}^k)$ interval. The second statement of the theorem states what happens if the change of $a_i$ crosses the boundary of an interval $(\bar{\alpha}^k, \ubar{\alpha}^k)$ , i.e., a change in the expected occupancy. Practically, the second statement of Theorem~\ref{theorem: E_derivative_a} states that if an impatience exits an $(\bar{\alpha}^k, \ubar{\alpha}^k)$ region there will be a jump or drop in the expected occupancy value. Hence, to find the expected occupancy value difference from this jump or drop in value one must take the difference just above and below $\ubar{\alpha}^k$. 

\subsection{Heterogeneity in the User Population}
\label{subsect: heterogeneity_user_pop}
In practice, it could be the case that the EV populations arriving at a charging facility have heterogeneous charging capabilities such that not all EVs arriving at a charging facility are free to choose any of the $L$ offered service levels \citep{afdc_energygov}. For example, consider the scenario where $L=3$ service levels are offered at a charging facility but only half of the population is able to choose service levels $\ell = 1$ or $2$ and the other half can only choose $\ell = 2$ or $3$. Furthermore, it is also possible and plausible that the different sub-populations of the EV charging facility have disparate impatience profiles. 

As a result of these phenomena, a charging facility is forced to estimate the probability an arriving user at an EV charging facility is a member of a sub-population that may only use a subset of the available chargers. In practice, a charging facility is at risk of incorrectly estimating the sub-population membership probabilities. Hence, we are motivated to study how the expected occupancy varies with varying sub-population estimate values. In the case where there exist heterogeneous user sub-populations that may only choose from a subset of the offered service levels we use the following assumption.

\begin{assump}
\label{assump: usersubpopulations}
 Users arriving at a charging facility are constrained to choose only from a subset of the offered service levels at the charging facility. Recall that $\mathcal{P}\left( \mathcal{L}\right)$ is the power set of the offered service levels minus the empty set and $C = 2^L - 1$ is the cardinality of this set. Furthermore, define $B_i \in \mathcal{P}\left( \mathcal{L}\right)$ to be user service-level subset $i$ where $i\in \{1, \dots, C\}$. Let $p_{B} = [p_{B_1}, \dots, p_{B_C}]^T$ such that $p_{B_i}$ is the sub-population probability and $\sum_{i=1}^{C} p_{B_i} = 1$.
 
\end{assump}

\begin{remark}
 We consider the possibility that a user sub-population may have a different impatience profile define $p^{B_i} = [p_{1}^{B_i}, \dots p_{M}^{B_i}]$ and  $a^{B_i} = [a_{1}^{B_i}, \dots a_{M}^{B_i}]$ to be the probability mass vector and impatience profile vector of subset $B_i$, respectively. Any quantity which is being computed with respect to a subset will be defined conditioned on some $B_i$, e.g., given $B_1 = \{1\}$ then $\E[\eta; p^{B_1}, a^{B_1}\mid B_1]$ is the expected occupancy conditioned on the user sub-population only being able to choose the slowest charging rate. If there is no conditioning on a sub-population, a quantity is for the case when all service levels are possible choices.
\end{remark}

Assumption \ref{assump: usersubpopulations} formally defines the occurrence of population subsets related to the possible charging choices for EV charging facility users.  Next, we  present Proposition \ref{proposition: grad_expect_pB_subpop} which proves the variability of the expected occupancy with respect to different estimate of the user sub-population probability. 

\begin{proposition}
\label{proposition: grad_expect_pB_subpop}
Consider a charging facility that offers $L$ service levels and is operating under Assumptions \ref{assump: rv_arrival}, \ref{assump: ordering_of_func}, \ref{assump: ratio_cannot_equal_a_i}, \ref{assump: timeatcharger}, and \ref{assump: usersubpopulations}. Let $p_{B} = [p_{B_1}, \dots, p_{B_C}]^T$ where $p_{B_i}$ is the sub-population probability of $B_i \in \mathcal{P}\left(\mathcal{L}\right)$ and $\sum_{i=1}^{C} p_{B_i} = 1$.
Then,
\begin{align*}
    \nabla_{p_B} \mathbb{E}[\eta] = \begin{bmatrix}
    \sum^L_{\ell = 1} \sum^{M}_{m=1}p_{m}^{B_1} \frac{\mathds{1}_{\ell }(a_{m}^{B_1} )}{R^\ell} \\ 
    \vdots \\
    \sum^L_{\ell = 1} \sum^{M}_{m=1}p_{m}^{B_C} \frac{\mathds{1}_{\ell}(a_{m}^{B_C})}{R^\ell}.
    \end{bmatrix}
\end{align*}
where $\mathds{1}_{\ell }(a_{m}^{B_i}) = 1$ if $\ubar{\alpha}^\ell < a_{m}^{B_i} < \bar{\alpha}^\ell$ and $\mathds{1}_{\ell}(a_{m}^{B_i}) = 0$ otherwise.

\end{proposition}

\begin{proof}
 From the law of total probability, we have that
\begin{align*}
    \mathbb{E}[\eta] = \sum^C_{i=1} \mathbb{E}[\eta \mid B_i]p_{B_i}.
\end{align*}
Realizing that $\mathbb{E}[\eta \mid B_i] = \sum^L_{\ell = 1} \sum^{M}_{m}p_{m}^{B_i} \frac{\mathds{1}_{\ell }(a_{m}^{B_i} )}{R^\ell}$ we can further simplify to
\begin{align*}
    \mathbb{E}[\eta] = \sum^C_{i=1} \left( \sum^L_{\ell = 1} \sum^{M}_{m}p_{m}^{B_i} \frac{\mathds{1}_{\ell }(a_{m}^{B_i} )}{R^\ell} \right)p_{B_i}.
\end{align*}
Lastly, take the gradient of $\E[\eta]$ with respect to $p_B$.
\end{proof}

\begin{example}
Consider a scenario where $L=2$ such that a charging facility offers only a slow and fast charging rate, and therefore there are $C = 3$ possible user sub-populations  and $\mathcal{P}\left(\mathcal{L}\right\} = \left\{B_1, B_2, B_3 \right\} = \left\{\{ 1\}, \{2\},\{1,2\}\right\} $ is the set of the sub-populations. Furthermore, from Assumption \ref{assump: usersubpopulations} recall $p_B = [p_{B_1}, p_{B_2}, p_{B_3}]$, and we suppose that $p_{B_2}= 0$ , i.e., no users are restricted to choosing only service level 2, and we remove this entry from $p_B$.

Suppose that the impatience profile of all the user sub-populations are the same such that there are $M = 2$ impatience profiles where $p = p^{B_1}  = p^{B_3} = [0.5, 0.5]$ and $a = a^{B_1}  = a^{B_3}$. Additionally, suppose $\ubar{\alpha}^1 < a_{1} < \bar{\alpha}^1$
and let $\ubar{\alpha}^2 < a_{2} < \bar{\alpha}^2$. 

Then, using Proposition \ref{proposition: grad_expect_pB_subpop}, we compute the gradient of the expected occupancy as
\begin{align}
    \label{example: gradient}
    \nabla_{p_B} \mathbb{E}[\eta] = \begin{bmatrix}
    \frac{0.5}{R^1} ,
    \frac{0.5}{R^2}
    \end{bmatrix}^\top.
\end{align}

In \eqref{example: gradient} we see that the expected occupancy varies linearly with $p_B$. Now consider the case where $p_{B_3} \rightarrow 0$ and $p_{B_1} \rightarrow 1$, i.e., the sub-population of users that use both charging levels disappears and the facility exclusively services users who can only choose the slowest service level. From \eqref{example: gradient} we see the expected occupancy will increase since the gradient for sub-population $B_1$ is greater than the gradient for $B_3$. 
\end{example}

Users arriving at an EV charging facility make the choice of service level based on their energy demand $x_j$ and impatience factor $\alpha_j$ as detailed in Section \ref{sect: problem_formulation}. The paper \citep{santoyo2020pricing} analyzes the EV charging problem with the assumption that users always remain until they receive a full charge. Consider the possibility that users are able to depart from the charging facility before the completion of their charging demand. In this scenario, arriving users make their service level choice, i.e., choice of charging rate and energy price, as before by solving \eqref{eq: level_func}. However, users now exhibit a new random behavior in the form of an early departure. This is formalized in the following assumption

\begin{assump}
\label{assump: define_partial_rv}
Define $\theta_j$ to be an i.i.d. random variable which defines the proportion of the desired charge $x_j$ that is received by user $j$. In particular, define the PDF of $\theta_j$ to be $f_\Theta(\theta)$ which has support on $[0, 1]$. Note that $x_j$ is independent of $\theta_j$. Furthermore, define $\theta_j x_j$ to be the amount of charge received by arriving user $j$. Then, user $j$ occupies a charger at the charging facility during the time interval $[\tau_j, \tau_j + \theta_j x_j/r_j]$.
\end{assump}


\begin{proposition}
\label{prop: early_departure}
Consider an EV charging facility operating under Assumptions \ref{assump: rv_arrival}, \ref{assump: ordering_of_func},  \ref{assump: ratio_cannot_equal_a_i}, and \ref{assump: define_partial_rv} such that users choose a service level according to \eqref{eq: level_func} with the possibility of a user only fulfilling $\theta_j x_j$ of charging demand at arrival. Then, we have that 
\begin{equation*}
    \frac{d\E[\eta;p,a]}{d\E[\theta]} = \lambda \E[x] \E[\frac{1}{r};p,a].
\end{equation*}
\end{proposition}
\begin{proof}
The proof follows immediately from the definition of the expected occupancy $\E[\eta; p,a] = \lambda \E[x] \E[\theta] \E[\frac{1}{r};p,a]$.
\end{proof}
From Proposition \ref{prop: early_departure}, a charging facility operator sees that the expected occupancy varies linearly with changes in the expected proportion of charge received by users at the charging facility. Furthermore, irrespective of the distribution of $\theta_j$, the expected occupancy will vary when the mean of the distribution varies.

\section{Numerical Study}
\label{sect: case_study}
In this section, we present a numerical study that illustrates different phenomena of which a charging facility operator must be mindful when choosing to provide a set of prices and charging rates. Specifically, we demonstrate practical occurrences of the results of Proposition~\ref{proposition:maxdiff_expectation} and~\ref{proposition: grad_expect_pB_subpop}, and Theorem~\ref{theorem: E_derivative_P} and~\ref{theorem: E_derivative_a}.\footnote{The Python code for this case study is available at \url{https://github.com/gtfactslab/evcharging_sensitivity}}. We first present the charging facility and user parameters in Section \ref{subsect: charging_facility_setup}. Then in Section \ref{subsect: mischaracterized_user_impatience} we consider the case when a charging facility is deciding between two pricing schemes such that the estimated expected occupancy is robust to mischaracterizations of the user impatience profiles. In this first case, no user sub-populations are considered, i.e., all arriving users can choose any of the offered service levels. Lastly in Section \ref{subsect: heterogeneous_pop_setup} we add the notion of user sub-populations  and study the effect of user choice limitations on their service level choice probabilities; similar to the first part, we identify the robust pricing scheme to uncharacterized homogeneity. 

\subsection{Charging Facility and User Parameters}
\label{subsect: charging_facility_setup}
Consider a charging facility offering $L = 3$ service levels that is deciding between offering two sets of prices and charging rates to meet this constraint. Specifically, this charging facility can either operate under the pricing scheme $A$ or $B$ whose parameters are detailed in Table \ref{tab:pricing_schemes_parameters}. 
\begin{table}[]
    \centering
    \caption{Pricing and Charging Rates for Pricing Schemes}
    \begin{tabular}{l@{\hspace{4pt}}c@{\hspace{4pt}}c}
          & Pricing Scheme A &  Pricing Scheme B\\\hline
    Charging Rates (kW.)   & $15,30,35$ & $15,30,35$  \\
    Prices  (\$/kWh.) & $0.15,0.25,0.32$  & $0.05, 0.25, 0.33$ \\ \hline
    \end{tabular}
    \label{tab:pricing_schemes_parameters}
\end{table}


In this particular case, a charging facility has a fixed set of charging rates to offer at each service level but is deciding what prices to charge for energy at each service level. At this charging facility user arrivals are a Poisson process with $\lambda = 30$ EVs/hr. Furthermore, user's energy demands (in kWh.) are distributed with $x_j \sim f_X(x_j)$ where $f_X(x_j) = \mathcal{U}(5, 100)$ where $\mathcal{U}(\cdot)$ denotes a uniform distribution.

The charging facility operator estimates user's impatience factor to be a PMF that is defined to be $p_A(\alpha; \tilde{p}, \tilde{a})$. We  consider the case when there are $M=4$ impatience profiles that describe the user population and suppose that the charging facility is correctly estimating that $M=4$ profiles also exist such that $\tilde{a} = [\tilde{a}_1, \tilde{a}_2, \tilde{a}_3, \tilde{a}_4]^\top$. 

While a charging facility may estimate the user impatience profiles to be discrete random variables with PMF parameters $\tilde{p}$ and $\tilde{a}$ this may not be an accurate assessment. Given the aforementioned prices and charging rates for pricing schemes $A$ and $B$ we consider the following two potential cases a charging facility may experience: first, a case when the true distribution of the impatience of users is in fact a discrete random variable with $p = \tilde{p}$ and $a \neq \tilde{a}$ and secondly when the true distribution is not a discrete random variable but in fact a bounded multi-modal normal distribution. In considering these cases, we are able to illustrate which of the two pricing schemes is robust to these types of errors by quantifying the error in the expected occupancy from the PMF estimate the facility has made and the two scenarios for the true distributions of the user impatience.  

First consider the case when a charging facility has estimated the user impatience PMF to have parameters $\tilde{p} = [0.25,0.25,0.25,0.25]^\top$ and $\tilde{a} = [2, 10, 20, 25]^\top$ but in reality the users impatience distributions are $p = \tilde{p}$ but $a = [2, 5.5, 20, 25]^\top$. The discrepancy between $a$ and $\tilde{a}$  will result in differing values for $\E[\eta; p, a]$ and $\E[\eta; \tilde{p}, \tilde{a}]$, respectively. 

We present Figure \ref{fig:discrete_discrete_impatience} which illustrates $p_A(\alpha_j; \tilde{p}, \tilde{a})$ and $p_A(\alpha_j; p, a)$ in the top and bottom plots, respectively. Furthermore, the respective values $(\ubar{\alpha}^k, \bar{\alpha}^k)$ computed for these two scenarios for both pricing schemes are displayed. Recall from Lemma \ref{lemma: probability_choose_level_special_case} that the probability mass of $p_A(\alpha_j; \tilde{p}, \tilde{a})$ or $p_A(\alpha_j; p, a)$ which falls in the non-empty intervals $(\ubar{\alpha}^k, \bar{\alpha}^k)$ determines the probability of users choosing service level $k$. As a result a discrepancy between $a$ and $\tilde{a}$ can lead to different expected occupancy between what the charging facility estimates and what actually occurs in practice. While problematic, a charging facility can mitigate such a discrepancy by ultimately choosing a pricing scheme that is resilient to such mischaracterizations. 

\subsection{Occupancy Error from Mischaracterized User Impatience}
\label{subsect: mischaracterized_user_impatience}
Given the pricing schemes $A$ and $B$ with $\tilde{a}$ we illustrate the intervals $(\ubar{\alpha}^k, \bar{\alpha}^k)$  in the plots of Figure \ref{fig:discrete_discrete_impatience}. Recall $S(x_j, \alpha_j)$ from \eqref{eq: level_func}, then the choices made by user $j$ in pricing schemes $A$ and $B$ are $S_A(x_j, \alpha_j)$ and $S_B(x_j, \alpha_j)$, respectively. Then, a charging facility with estimate $\tilde{a}$ concludes that $\P(S_A(x_j, \alpha_j) = 1) = \P(S_A(x_j, \alpha_j) = 2) = 0.25$ and $\P(S_A(x_j, \alpha_j) = 3) = 0.50$, and $\P(S_B(x_j, \alpha_j) = 1) = \P(S_B(x_j, \alpha_j) = 2) = 0.25$ and $\P(S_B(x_j, \alpha_j) = 3) = 0.50$. 

In reality, the true impatience values $a$ demonstrated in the bottom plot of Figure \ref{fig:discrete_discrete_impatience} show that $\P(S_A(x_j, \alpha_j) = 1) = \P(S_A(x_j, \alpha_j) = 2) = 0.25$ and $\P(S_A(x_j, \alpha_j) = 3) = 0.50$ and that $\P(S_B(x_j, \alpha_j) = 1) = 0.50 $, $\P(S_B(x_j, \alpha_j) = 2) = 0$ and $\P(S_B(x_j, \alpha_j) = 3) = 0.50$. In practice, for pricing scheme $B$ this states that even though a charging facility is offering 3 service levels, only 2 of them will be chosen by users. Numerically, this leads to an over $20\%$ error between $\E[\eta; p,a]$ and $\E[\eta; \tilde{p},\tilde{a}]$ when a charging facility chooses pricing scheme $B$ due to the charging facility  underestimating the expected occupancy. However, in this case, no error arises when using pricing scheme $A$.

While this is illustrative, this represents the considerations a charging facility must make when setting energy prices or charging rates to make their pricing scheme resilient to mischaracterizations of users' impatience. In analyzing the variability of the expected occupancy when utilizing $\tilde{a}$ or $a$ we see an illustration of Statement~1 of Theorem~\ref{theorem: E_derivative_a}. Specifically, if $a_1 > \bar{\alpha}^1$  for pricing scheme $B$ then Statement~1 of Theorem \ref{theorem: E_derivative_a} would predict that the expectation would have remained the same. However, since $a_1 < \bar{\alpha}^1$ for pricing scheme $B$ the difference in the expectation is as predicted in Statement 2 of Theorem \ref{theorem: E_derivative_a}. We illustrate the variation of the expected occupancy with varying $\tilde{a}$ for both pricing schemes in Figure \ref{fig:step_plot} when both the estimated and true impatience are discrete random variables. Lastly, consider the hypothetical scenario where $\tilde{a} = a$ but $\tilde{p} \neq p$, then we see from analyzing Figure \ref{fig:discrete_discrete_impatience} that $\E[\eta; p,a]$ will vary linearly with $p$.

\begin{figure}
    \centering
    \begin{tikzpicture}
    \begin{axis}
    [title=User Service Level Choice Regions \\ with Estimated Impatience,
    title style={align = center, yshift=1.5ex},
     xlabel={$\alpha_j$ [\$/hr.]},
    xmin = 0,
    xmax = 28,
    xtick={0,5,10,15,20,25},
    ymin= 0,
    ymax=0.3,
    width=7cm,
    height=3cm,
    clip=false,
    ylabel={$p_A(\alpha_j; \tilde{p}, \tilde{a})$},
    ylabel near ticks,
     set layers,
    axis on top,
    scale only axis]
    \addplot+[ycomb, red!50, mark options={red!50}, line width=0.5mm] plot coordinates
    	{(2,0.25) (10,0.25) (20,0.25) (25,0.25)}; 
    	
    \draw [dashed, color=mycolor1, line width=0.25ex] (axis cs:3.0,0) -- ++(axis cs:0,0.3) node[above] {$\bar{\alpha}^1 $};
    
    \draw [dashed, color=mycolor1, line width=0.25ex] (axis cs:14.7,0) -- ++(axis cs:0,0.3) node[above] {$\bar{\alpha}^2 $};
    \draw [dotted, line width=0.25ex, color=mycolor2] (axis cs:6.0
    ,0) -- ++(axis cs:0,0.3) node[above] {$\bar{\alpha}^1 $};
    
    \draw [dotted, line width=0.25ex,color=mycolor2] (axis cs:16.80,0) -- ++(axis cs:0,0.3) node[above] {$\bar{\alpha}^2 $};

    \end{axis}
    \end{tikzpicture}
    
    \vspace{0.5em}
    \begin{tikzpicture}
    \begin{axis}[title=User Service Level Choice Regions \\ with True Impatience, 
    title style={align = center, yshift=1.5ex},
    xlabel={$\alpha_j$ [\$/hr.]},
    xmin = 0,
    xmax = 28,
    xtick={0,5,10,15,20,25},
    ymin= 0,
    ymax=0.3,
    width=7cm,
    height=3cm,
    clip=false,
    ylabel={$p_A(\alpha_j; p, a)$},
    ylabel near ticks,
     set layers,
    axis on top,
    legend columns=2,
    scale only axis,
    legend style={at={(0.5,-0.4)},anchor=north}]
    \addplot+[ycomb, ycomb, red!50, mark options={red!50}, line width=0.5mm, forget plot] plot coordinates
    	{(2,0.25) (5.5,0.25) (20,0.25) (25,0.25)}; 
    	
    \draw [dashed, color=mycolor1, line width=0.25ex] (axis cs:3.0,0) -- ++(axis cs:0,0.3) node[above] {$\bar{\alpha}^1 $};
    
    \draw [dashed, color=mycolor1, line width=0.25ex] (axis cs:14.7,0) -- ++(axis cs:0,0.3) node[above] {$\bar{\alpha}^2 $};
    \draw [dotted, line width=0.25ex, color=mycolor2] (axis cs:6.0
    ,0) -- ++(axis cs:0,0.3) node[above] {$\bar{\alpha}^1 $};
    
    \draw [dotted, line width=0.25ex,color=mycolor2] (axis cs:16.80,0) -- ++(axis cs:0,0.3) node[above] {$\bar{\alpha}^2 $};
    
    \addlegendimage{dashed, color=mycolor1,line width=0.25ex};
    \addlegendentry{Pricing Scheme A};
    \addlegendimage{dotted, line width=0.25ex,color=mycolor2};
    \addlegendentry{Pricing Scheme B};
    \end{axis}

    \end{tikzpicture}

    \caption{A charging facility 
    decides between pricing scheme A or B when it has estimated the user impatience as in the top plot of this figure. In reality, we suppose the user impatience is as shown in the bottom plot of this figure. The discrepancy between the estimated impatience values $\tilde{a}$ and the true impatience values $a$ leads to the charging facility having an incorrect estimate of the expected occupancy. To mitigate this, a charging facility can choose its prices and charging rates such that their estimate of the expected occupancy is resilient to impatience mischaracterizations. As described in Section \ref{sect: case_study}, pricing scheme A leads to the same expected occupancy under both the true and estimated impatience, while pricing scheme B leads to expected occupancy that is over 20\% larger for the true impatience distribution compared to the estimated impatience distribution.}
    \label{fig:discrete_discrete_impatience}
\end{figure}
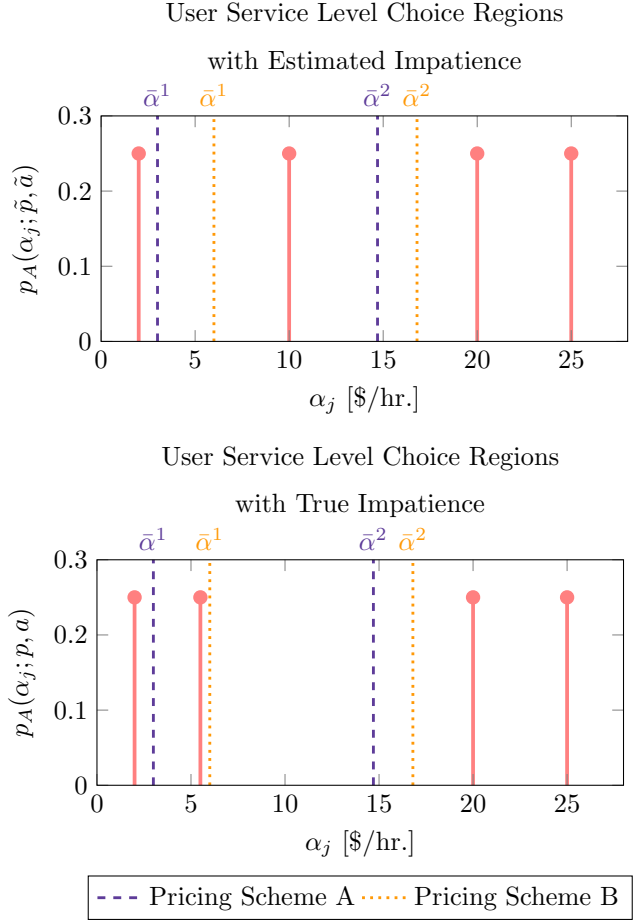

\begin{figure}
    \centering
    \begin{tikzpicture}
    
    \begin{axis}[
    title={Expected Occupancy for Pricing Scheme $A$ and $B$ \\ with Varying Impatience},
    title style={align = center, yshift=1.5ex},
    ylabel={$\mathbb{E}[\eta; \tilde{p},\tilde{a}]$},
    ylabel near ticks,
    xlabel={$\tilde{a}_2$ [\$/hr.]},
    xmin = 2,
    xmax = 20,
    ymin = 55,
    ymax= 80,
    width=7cm,
    height=4cm,
    legend columns=2,
    scale only axis,
    legend style={at={(0.5,-0.5)},anchor=north}
    ]
    \addplot[mark=none, color=mycolor1, line width = 0.25ex] table {plotdata/E_A.txt};
    \addlegendentry{Pricing Scheme $A$};

    \addplot[dashed, mark=none, color=mycolor2,  line width = 0.35ex] table {plotdata/E_B.txt};
    \addlegendentry{Pricing Scheme $B$};

    \end{axis}
    \end{tikzpicture}
    \caption{An EV charging facility where users arrive according to a Poisson process with random demand and impatience decides between a Pricing Scheme $A$ and $B$. At arrival, users choose a charging rate from a collection of service levels that minimizes the total cost to themselves that includes their impatience. A charging facility has estimated the impatience PMF for the arriving users and uses this estimate to compute the expected occupancy for both pricing schemes. In this plot, we illustrate the variation in the expected occupancy when $\tilde{a} = [2, \tilde{a}_2, 20, 25]^\top$. As a charging facility's estimate of $\tilde{a}_2$ varies so does the expected occupancy under each pricing scheme.}
    \label{fig:step_plot}
\end{figure}
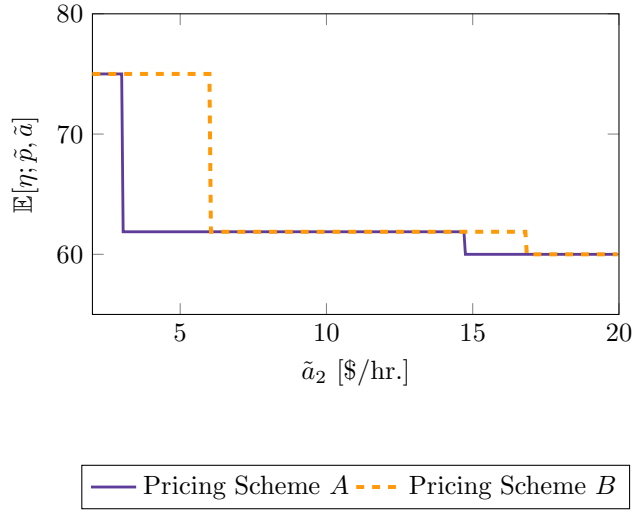

Secondly, we consider the case when a charging facility has estimated the impatience to be a discrete random variable as before but in reality the true user impatience distribution is a truncated multi-modal normal distribution $f_A(\alpha_j)$. Given the multi-modal normal distribution illustrated in Figure \ref{fig:discrete_continuous_impatience} we have that $\P(S_A(x_j, \alpha_j) = 1) = 0.221$, $\P(S_A(x_j, \alpha_j) = 2) = 0.281 $ $\P(S_A(x_j, \alpha_j) = 3) = 0.498$ and that $\P(S_B(x_j, \alpha_j) = 1) = 0.431 $, $\P(S_B(x_j, \alpha_j) = 2) = 0.070$ and $\P(S_B(x_j, \alpha_j) = 3) = 0.499$. Numerically, this leads to an over $15\%$ error in the expected occupancy between the estimated and true user impatience 
when a charging facility chooses pricing scheme $B$ due to the charging facility underestimating the expected occupancy. Furthermore, pricing scheme $A$ now has an approximately $3\%$ error from the true value where the charging facility has overestimated the expected occupancy. The increase in error in pricing scheme $A$ in this scenario is due to the portion from the first mode of $f_A(\alpha_j)$ that is in $(\ubar{\alpha}^2, \bar{\alpha}^2)$, i.e., $(\bar{\alpha}^1, \bar{\alpha}^2)$, of pricing scheme $A$ and the portion of the second mode of $f_A(\alpha_j)$ from the first mode that is in $(\ubar{\alpha}^2, \bar{\alpha}^2)$, i.e., $(\bar{\alpha}^1, \bar{\alpha}^2)$, of pricing scheme $B$. 

We present the numerical values from this part of the numerical study in Table \ref{tab:expect_num_vehicles_cases}. In both the scenarios where the true PMF is a discrete random variable and the one where the true distribution of the impatience is a probability density function we see that pricing scheme $A$ is more resilient to mischaracterizations of $p_A(\alpha_j; \tilde{p}, \tilde{a})$ (or $f_A(\alpha_j)$) than pricing scheme $B$. In fact, pricing scheme $B$ leads the charging facility operator to estimate a lower occupancy when in reality the value is higher. This potentially overburdens the facility's space resources. In practice, an operator gains insightful information from the operation of a charging facility on users' impatience that will guide them in choosing mischaracterization-resilient prices and charging rates. 

\subsection{Heterogeneous Charging Populations}
\label{subsect: heterogeneous_pop_setup}
Similar to Section \ref{subsect: mischaracterized_user_impatience} we study the sensitivity of the expected occupancy at the charging facility to the arrival of heterogeneous charging populations when the impatience profiles are and are not mischaracterized. Specifically, the charging facility evaluates the robustness of pricing schemes $A$ and $B$ defined in Table \ref{tab:pricing_schemes_parameters} to the presence of heterogeneous populations.

Recall Assumption \ref{assump: usersubpopulations}, then given $L = 3$, $C = 7$ is the total number of possible choice subset combinations for the offered service levels and  $\mathcal{P}\left(\mathcal{L}\right) = \left\{\{1\}, \{2\}, \{3\}, \{1,2\}, \{1,3\}, \{2,3\}, \{1,2,3\} \right\}$. Define $\tilde{p}_B = [0,0,0,0,0,0,1]$ to be the $\mathcal{P}\left(\mathcal{L}\right)$ sub-population probability vector estimated by the charging facility. In practice, this means that the charging facility is expecting all of the arriving users at the charging facility to be able to use all of the offered charging levels; however, this may not be the case due to the technological limitations of specific EVs \citep{afdc_energygov}.

In reality, any of the arriving users at this charging facility may only be able to charge with any of the charging level combinations $B_i \in \mathcal{P}(\mathcal{L})$. Consider the case where not all users may choose all of the offered charging levels; specifically, define the true sub-population probabilities to be $p_B = [0.25, 0, 0, 0.25, 0, 0, 0.50]$, i.e., one fourth of the population can select $B_1 = \{1\}$, one fourth can select $B_4 = \{1,2\}$ and the remaining half can select $B_7 = \{1,2,3\}$. 

A charging facility can determine the robustness of one set of charging rates to variability in the expected occupancy with heterogeneous charging populations and mischaracterized impatience profiles. Since the charging facility has estimated that all users can make use of all of the offered charging rates Figure \ref{fig:discrete_discrete_impatience} illustrates the regions of choice for both pricing scheme $A$ and $B$ along with the estimated impatience $\tilde{p}, \tilde{a}$ and estimated sub-population probabilities $\tilde{p}_{B}$. Specifically, a charging facility computes that  $\P(S_A(x_j, \alpha_j) = 1 \mid B_7) = \P(S_A(x_j, \alpha_j) = 2 \mid B_7) = 0.25$ and $\P(S_A(x_j, \alpha_j) = 3\mid B_7) = 0.50$, and $\P(S_B(x_j, \alpha_j) = 1\mid B_7) = \P(S_B(x_j, \alpha_j) = 2\mid B_7) = 0.25$ and $\P(S_B(x_j, \alpha_j) = 3\mid B_7) = 0.50$. Note that since the charging facility is operating under the perception that all of the arriving users can choose from any of the offered service levels the conditioning on $B_7$ can be dropped.

In Figure \ref{fig: heterogeneous_true} we see the service level choice regions for the true impatience profile $p, a$ and true sub-population probabilities $p_B$. Notice in Figure \ref{fig: heterogeneous_true} that the choice region for service level 2 has disappeared in the left plot.  
For pricing scheme A, we have $\P\left(S_A(x_j,\alpha_j) = 1 \mid B_1 \right) = 1.0$, $\P\left(S_A(x_j,\alpha_j) = 1 \mid B_4 \right) = 0.25$,  $\P\left(S_A(x_j,\alpha_j) = 2 \mid B_4 \right) = 0.75$, $\P\left(S_A(x_j,\alpha_j) = 1 \mid B_7 \right) = \P\left(S_A(x_j,\alpha_j) = 2 \mid B_7 \right) = 0.25$ and $\P\left(S_A(x_j,\alpha_j) = 3 \mid B_7 \right) = 0.50$. Then, using the law of total probability one computes that $\P\left(S_A(x_j,\alpha_j) = 1 \right) = 0.4375$, $\P\left(S_A(x_j,\alpha_j) = 2 \right) = 0.3125$, and $\P\left(S_A(x_j,\alpha_j) = 3 \right) = 0.25$. 

For pricing scheme B, we have that $\P(S_B(x_j,\alpha_j) = 1 \mid B_1 ) = 1.0$, $\P(S_B(x_j,\alpha_j) = 1 \mid B_4) = 0.50$,  $\P\left(S_B(x_j,\alpha_j) = 2 \mid B_4 \right) = 0.50$, $\P(S_B(x_j,\alpha_j) = 1 \mid B_7) =0.50 $, $\P\left(S_B(x_j,\alpha_j) = 2 \mid B_7 \right) = 0$ and $\P\left(S_B(x_j,\alpha_j) = 3\mid B_7\right)=0.50$. Then, for pricing scheme B using the law of total probability we have that $\P\left(S_B(x_j,\alpha_j) = 1 \right) = 0.625$, $\P\left(S_B(x_j,\alpha_j) = 2 \right) = 0.125$, and $\P\left(S_B(x_j,\alpha_j) = 3 \right) = 0.25$.

Given these parameters and probabilities, a charging facility studies how each pricing scheme performs in terms of leading to a correct estimate of the expected occupancy. Specifically, pricing scheme $A$ has an approximate error of approximately $15\%$ while pricing scheme $B$ has an under-approximation error of approximately $25\%$. In this particular scenario we see that the prices and rates offered by pricing scheme $A$ are more robust to mischaracterizations for both the users impatience and sub-population probabilities. 

In a similar fashion, one can study the a scenario where a charging facility operates with correct knowledge of the user impatience $p$ and $a$ while having $\tilde{p}_B \neq p_B$. In this scenario, pricing scheme $A$ has an expected value error of approximately $15\%$ which pricing scheme $B$ has an approximate error of $10\%$. In this case, pricing scheme $B$ is more robust to this particular mischaracterization of the user sub-populations.

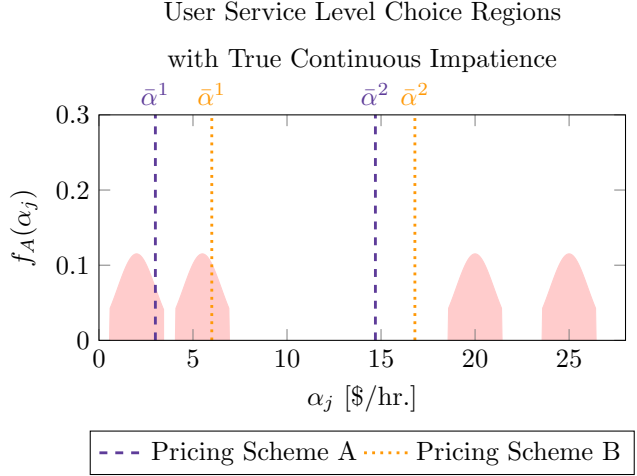
\begin{figure}
    \centering
    \begin{tikzpicture}
    \begin{axis}
    [title={User Service Level Choice Regions \\ with True  Continuous Impatience},
    title style={align=center, yshift=1.5ex},
    xlabel={$\alpha_j$ [\$/hr.]},
    xmin = 0,
    xmax = 28,
    xtick={0,5,10,15,20,25},
    ymin= 0,
    ymax=0.3,
    width=7cm,
    height=3cm,
    clip=false,
    ylabel={$f_A(\alpha_j)$},
    ylabel near ticks,
    set layers,
    axis on top,
    scale only axis,
    legend columns=2,
    legend style={at={(0.5,-0.4)},anchor=north}
    ]

    \addplot+[mark=none, area legend, color=red!20, fill=red!20, forget plot] table[x index=0, y expr={\thisrowno{1}/4}] {plotdata/pmf0.txt}; 
    	
    \addplot+[mark=none, area legend, color=red!20, fill=red!20,, forget plot] table[x index=0, y expr={\thisrowno{1}/4}] {plotdata/pmf1.txt}; 
    
    \addplot+[mark=none, area legend, fill=blue, color=red!20, fill=red!20,, forget plot] table[x index=0, y expr={\thisrowno{1}/4}] {plotdata/pmf2.txt}; 
    
    \addplot+[mark=none, area legend, color=red!20, fill=red!20] table[x index=0, y expr={\thisrowno{1}/4}, forget plot] {plotdata/pmf3.txt}; 
    \draw [dashed, color=mycolor1,line width=0.25ex] (axis cs:3.0,0) -- ++(axis cs:0,0.3) node[above] {$\bar{\alpha}^1 $};
    
    \draw [dashed, color=mycolor1, line width=0.25ex] (axis cs:14.7,0) -- ++(axis cs:0,0.3) node[above] {$\bar{\alpha}^2 $};
    \draw [dotted, line width=0.25ex, color=mycolor2] (axis cs:6.0
    ,0) -- ++(axis cs:0,0.3) node[above] {$\bar{\alpha}^1 $};
    
    \draw [dotted, line width=0.25ex,color=mycolor2] (axis cs:16.80,0) -- ++(axis cs:0,0.3) node[above] {$\bar{\alpha}^2 $};

    \addlegendimage{dashed, line width=0.25ex,color=mycolor1};
    \addlegendentry{Pricing Scheme A};
    \addlegendimage{dotted, line width=0.25ex,color=mycolor2};
    \addlegendentry{Pricing Scheme B};
    \end{axis}
    \end{tikzpicture}

    \caption{A charging facility decides between implementing pricing scheme $A$ or $B$ when it has estimated the user impatience as in the top plot of Figure \ref{fig:discrete_discrete_impatience}. In reality, we suppose the user impatience is a truncated multi-modal normal distribution as is shown in this plot. Hence, a charging facility has a mischaracterized value of the expected occupancy at the charging facility. The discrepancy between the estimated impatience values $\tilde{a}$ and the true impatience value $a$ leads to the charging facility having an incorrect estimate of the expected occupancy at a charging facility. Similar to the case in Figure \ref{fig:discrete_discrete_impatience}, Pricing Scheme A leads to true expected occupancy that is within 3\% of the estimated expected occupancy, while Pricing Scheme B leads to true expected occupancy that is over 15\% higher than estimated.}
    \label{fig:discrete_continuous_impatience}
\end{figure}

\begin{table}[]
    \centering
    \caption{Expected Occupancy for Different Probability Distributions and Pricing Schemes}
    \begin{tabular}{l@{\hspace{4pt}}c@{\hspace{4pt}}c}
          &  Expected Occupancy & Expected Occupancy \\
         Probability Distribution &Pricing Scheme A &  Pricing Scheme B\\\hline
    Estimated Impatience    & $61.875$ & $61.875$ \\
    True Discrete Impatience  & $61.875$ & $75.0$ \\
    True Multi-Modal Impat.   &  $60.38$ & $71.40$ \\ \hline
    \end{tabular}
    \label{tab:expect_num_vehicles_cases}
\end{table}



\begin{figure}
    \centering
     \center{\textbf{User Service Level Choice Regions with True Impatience and Heterogeneous User Populations}}
     \par\medskip
    \begin{tikzpicture}
    \begin{groupplot}[
        group style={
            group name=my plots,
            group size=2 by 1,
            xlabels at=edge bottom,
            ylabels at=edge left,
            horizontal sep=1cm,
            vertical sep=2cm
            },
        xlabel={$\alpha_j$ [\$/hr.]},
        xmin = 0,
        xmax = 28,
        xtick={0,5,10,15,20,25},
        width=0.5\linewidth,
        ymin= 0,
        ymax=0.3,
        clip=false,
        ylabel={$f_A(\alpha_j)$},
        ylabel near ticks,
        set layers,
        axis on top,
        legend columns=2,
        legend style={at={(-0.25,-0.5)},anchor=north},
    ]
        \nextgroupplot[title={$B_4 = \{1,2\}$, $p_{B_4} = \frac{1}{4}$}, title style={yshift=0.15cm, align=center}]
        \addplot+[ycomb, red!50, mark options={red!50}, line width=0.5mm] plot coordinates
    	{(2,0.25) (5.5,0.25) (20,0.25) (25,0.25)}; 
    	
        \draw [dashed, color=mycolor1,line width=0.25ex] (axis cs:3.0,0) -- ++(axis cs:0,0.3) node[above] {$\bar{\alpha}^1 $};
        
        \draw [dotted, line width=0.25ex, color=mycolor2] (axis cs:6.0
        ,0) -- ++(axis cs:0,0.3) node[above] {$\bar{\alpha}^1 $};
        
        \nextgroupplot[title={$B_7 = \{1,2,3\}$, $p_{B_7} = \frac{1}{2}$}, title style={yshift=0.15cm, align=center}]
        \addplot+[ycomb, red!50, mark options={red!50}, line width=0.5mm, forget plot] plot coordinates
    	{(2,0.25) (5.5,0.25) (20,0.25) (25,0.25)}; 
    	
        \draw [dashed, color=mycolor1,line width=0.25ex] (axis cs:3.0,0) -- ++(axis cs:0,0.3) node[above] {$\bar{\alpha}^1 $};
        
        \draw [dashed, color=mycolor1, line width=0.25ex] (axis cs:14.7,0) -- ++(axis cs:0,0.3) node[above] {$\bar{\alpha}^2 $};
        \draw [dotted, line width=0.25ex, color=mycolor2] (axis cs:6.0
        ,0) -- ++(axis cs:0,0.3) node[above] {$\bar{\alpha}^1 $};
        
        \draw [dotted, line width=0.25ex,color=mycolor2] (axis cs:16.80,0) -- ++(axis cs:0,0.3) node[above] {$\bar{\alpha}^2 $};
        
    \addlegendimage{dashed, line width=0.25ex,color=mycolor1};
    \addlegendentry{Pricing Scheme A};
    \addlegendimage{dotted, line width=0.25ex,color=mycolor2};
    \addlegendentry{Pricing Scheme B};
    
    \end{groupplot}
\end{tikzpicture}

    \caption{A charging facility offering $L = 3$ charging levels to arriving users who may not be able to choose from the $L$ offered service levels and whose impatience values may be incorrectly approximated must choose prices and charging rates in a way that the estimated expected occupancy at the charging facility is near the true expected occupancy. In practice, a charging facility must estimate the arriving user's impatience $\tilde{p}$ and the probability the arriving users which can use a subset of the offered service level charging rates $\tilde{p}_B$. In Section \ref{subsect: heterogeneity_user_pop}, we study the effects of the presence of heterogeneous charging populations in addition to mischaracterized user impatience profiles. In this plot, when $p_B = [0.25, 0, 0, 0.25, 0, 0, 0.5]$ we see the how the placement of the impatience mass couples with the user charging population heterogeneity. We see that pricing scheme A leads to an expected occupancy error of around $15\%$ while pricing scheme B has an expected occupancy error of around $25\%$.}
    \label{fig: heterogeneous_true}
\end{figure}
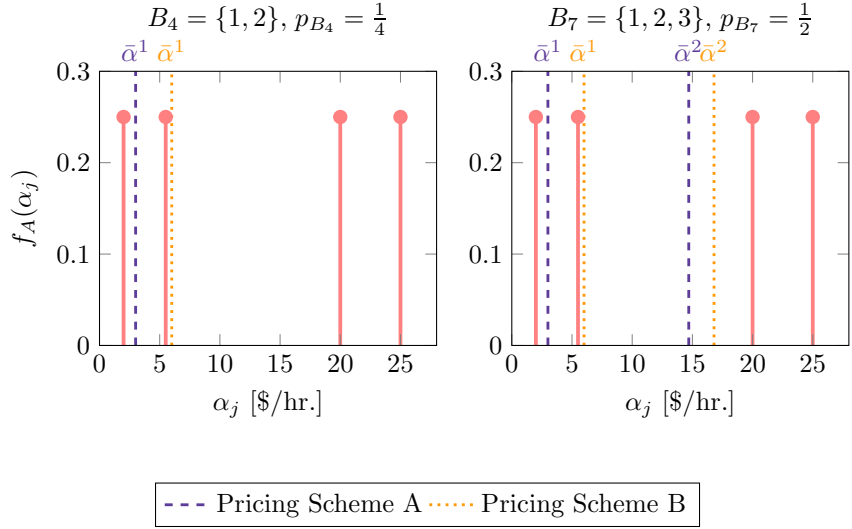

\section{Conclusion}
\label{sect: conclusion}
In this paper, we studied the problem of a charging facility operating with a defined service level model where users arrive randomly with a collection of random parameters. We specifically focus on the case where a charging facility is primarily interested in characterizing the expected occupancy at the charging facility. To compute the expected occupancy, a charging facility uses its knowledge on the distributions of the user arrivals and the respective parameters (energy demand and impatience factor). While useful, these computations are vulnerable to incorrect assessments by the charging facilities of the distribution of user parameters. Specifically, within the model, computing the expected occupancy is highly dependent on having the correct knowledge of the distribution of the user's impatience. As a result, we study the variability in the expected occupancy when the distribution and values of the impatience factor are mischaracterized. Furthermore, we also compute a worst-case error bound for the expected occupancy when the impatience factor is mischaracterized. Lastly, we consider the effects of user population heterogeneity, e.g., not all users being able to use all available charging levels, on the expected occupancy. We study the analytical results via a numerical study that illustrates how a charging facility operator can intelligently set prices and charging rates. While this paper studies the application of an opportunity cost-based user choice model to analyze electric vehicle charging facilities the results are applicable in many settings, where users arrive with demands and perform a personal trade-off between service rate and cost. Cloud computing \citep{chen2019pricing} and ride-sharing services \citep{banerjee2015pricing} among others are examples of when users can face similar trade-off decisions when paying for a service \citep{allon2008service}.

\bibliography{references.bib}  

\end{document}